\documentclass[journal]{IEEEtran}
\usepackage{tabularx}
\usepackage{longtable}
\usepackage{enumitem}
\usepackage{amsmath,amssymb,amsthm}
\usepackage{mathtools}
\usepackage{graphicx}
\usepackage{cite}
\usepackage{url}
\usepackage{algorithm}
\usepackage{algpseudocode}
\usepackage{booktabs}
\usepackage{tabularx}
\usepackage{booktabs}
\usepackage{comment}

\newtheorem{theorem}{Theorem}

\begin{document}

\title{Learning over Forward-Invariant Policy Classes: Reinforcement Learning without Safety Concerns}

\author{ 
Chieh~Tsai,
Muhammad~Junayed~Hasan~Zahed,
Salim~Hariri,
and~Hossein~Rastgoftar%
}

\maketitle

\begin{abstract}
This paper proposes a safe reinforcement learning (RL) framework based on
forward-invariance-induced action-space design. The control problem is cast as
a Markov decision process, but instead of relying on runtime shielding or
penalty-based constraints, safety is embedded directly into the action
representation. Specifically, we construct a finite admissible action set in
which each discrete action corresponds to a stabilizing feedback law that
preserves forward invariance of a prescribed safe state set. Consequently, the
RL agent optimizes policies over a safe-by-construction policy class. We
validate the framework on a quadcopter hover-regulation problem under
disturbance. Simulation results show that the learned policy improves
closed-loop performance and switching efficiency, while all evaluated policies
remain safety-preserving. The proposed formulation decouples safety assurance
from performance optimization and provides a promising foundation for safe
learning in nonlinear systems.
\end{abstract}

\begin{IEEEkeywords}
Safe reinforcement learning,
Markov decision process,
Forward invariance,
Gain scheduling,
Nonlinear control,
Quadcopter hover regulation,
Safety-critical systems
\end{IEEEkeywords}

\section{Introduction}

Reinforcement learning (RL) has created new opportunities for adaptive control
of nonlinear dynamical systems, especially when operating conditions vary over
time and fixed feedback laws become conservative. However, in safety-critical
systems such as quadcopters, autonomous vehicles, and robotic platforms, a
central difficulty remains: the policy must improve performance without
violating hard safety requirements during either training or deployment.
Existing approaches often address this tension by augmenting the learning
architecture with penalties, shields, barrier constraints, or online safety
filters. Although effective in many settings, such mechanisms typically treat
safety as an external correction layer rather than as an intrinsic property of
the policy class.

This issue is particularly important for quadcopter control. Nonlinear flight
dynamics, underactuation, actuator limits, and attitude constraints make aggressive adaptive control challenging, especially when the vehicle transitions
between large transient errors and near-hover operation. High-performance model-based controllers can provide strong stability and tracking guarantees, but they usually rely on gains selected offline. As a result, the same feedback configuration must simultaneously handle fast transient rejection, hover regulation, and robustness to disturbances, which can lead to conservative performance across different operating regimes.

This paper develops a safe RL framework in which safety is embedded directly into the action representation. We model the closed-loop decision problem as a Markov decision process (MDP), but instead of allowing arbitrary control actions, we construct a finite admissible action set whose elements correspond to pre-certified stabilizing feedback laws. Each admissible action preserves forward invariance of a prescribed safe state set, so every policy defined over this action space is safety-preserving by construction. Consequently, learning is performed over a forward-invariant policy class rather than over an unconstrained control class that must later be corrected. In this paper, the proposed framework is instantiated on a quadcopter hover-regulation problem.
Organization of the paper is as follows.  Section~\ref{Problem Statement} presents the problem statement. Section~\ref{Case Study: Quadcopter Hover Regulation} introduces the quadcopter hover-regulation case study and the forward-invariance-preserving control design. Section~\ref{DQN-based Safe Learning} presents the DQN-based safe learning framework. Section~\ref{subsec:sim_setup} provides the simulation results. Finally, Section~\ref{Conclusion} concludes the paper.

\subsection{Related Work}

Quadcopter trajectory tracking has been studied extensively using nonlinear,
model-based control methods that exploit the structure of rigid-body flight
dynamics. Foundational works established the modeling and low-level control
framework for miniature quadrotors, including indoor platform development,
PID/LQ comparisons, and backstepping-based nonlinear control
\cite{Bouabdallah2004ICRA,Bouabdallah2004IROS,Castillo2004TCST,Madani2006IROS}.
Subsequent experimental and theory-driven studies clarified the effects of
underactuation, actuator dynamics, and full-pose stabilization in practical
quadrotor control architectures
\cite{Hoffmann2007AIAA,Bouabdallah2007IROS,Nagaty2013JIRS}. More recent
geometric and differential-flatness formulations provide strong tracking
guarantees while avoiding Euler-angle singularities and enabling aggressive,
dynamically feasible flight
\cite{Lee2010CDC,Mellinger2011ICRA,Faessler2018RAL}. Although these approaches
have achieved excellent performance, they are typically deployed with fixed
feedback gains selected offline, which may become conservative across
substantially different operating regimes such as large transients, aggressive
maneuvers, and near-hover flight.

Gain scheduling offers a natural mechanism for adapting controller
aggressiveness across changing operating conditions. In the quadrotor
literature, gain-scheduled PID designs have been investigated for
fault-tolerant control and path tracking under actuator degradation or regime
variation \cite{Milhim2010AIAA,Amoozgar2012IFAC}. These works demonstrate that
scheduled gains can improve closed-loop performance relative to a single
operating-point controller, especially when the vehicle experiences parameter
changes, faults, or markedly different maneuver envelopes. However, most
existing gain-scheduling methods rely on heuristic interpolation, fuzzy
supervision, fault logic, or manually tuned switching rules
\cite{Milhim2010AIAA,Amoozgar2012IFAC}. Consequently, they generally do not
provide a formal guarantee that the online gain updates preserve the safety
envelope of the nonlinear closed-loop system.

This issue is closely tied to forward invariance. In constrained nonlinear
control, forward invariance requires that trajectories initialized in an
admissible set remain in that set for all future time
\cite{Blanchini1999Automatica}. Barrier-certificate methods established an
important verification framework for safety with respect to unsafe sets
\cite{Prajna2007TAC}, while control barrier function formulations provided
constructive real-time tools for enforcing invariant safe sets online
\cite{Ames2017TAC}. Related safe-learning work has also emphasized certified
regions of attraction, safe set expansion, and model-based stability
guarantees for uncertain nonlinear systems
\cite{Berkenkamp2016CDC,Berkenkamp2017NeurIPS}. For quadcopters, this
viewpoint is particularly important because position, attitude, angular-rate,
and thrust constraints are not merely soft performance targets; violating them
can compromise feasibility and invalidate controller assumptions. These
limitation motivate the use of admissible gain sets whose elements are certified
to preserve forward invariance of the closed-loop quadcopter dynamics.

From the reinforcement-learning viewpoint, the sequential decision-making
problem is naturally modeled as a Markov decision process (MDP), which
provides a standard framework for state evolution, action selection, and
long-horizon reward optimization
\cite{Bellman1957,Puterman1994,SuttonBarto2018}. This abstraction is
especially useful in control because it separates the system dynamics from the
policy class used for decision making. In safe RL, the MDP viewpoint is often
extended through constrained or safety-aware formulations, where return
maximization is supplemented by additional safety requirements
\cite{Achiam2017ICML,Chow2018NeurIPS}. Nevertheless, in many such approaches,
safety is enforced through runtime projection, auxiliary penalties, Lyapunov
constraints, or barrier-function filtering, rather than by restricting the
action space itself to a set of safety-certified controllers.

Learning-based flight control has emerged as an attractive alternative for
reducing hand tuning and compensating for modeling errors. Reinforcement
learning has been applied to low-level UAV attitude control and has shown
promising performance in high-fidelity simulation \cite{Koch2019TCPS}. More
broadly, deep Q-learning demonstrated the effectiveness of value-based RL over
discrete action sets \cite{Mnih2015Nature}. Safe RL has further introduced
constrained policy optimization, Lyapunov-based updates, model-based
stability-certified learning, and barrier-function-based safety mechanisms
\cite{Achiam2017ICML,Berkenkamp2017NeurIPS,Chow2018NeurIPS,Cheng2019AAAI}.
However, most RL-based quadrotor controllers still learn control actions
directly, or else adapt controller parameters without embedding formal
invariance guarantees into the gain-selection mechanism.

In contrast, the present work formulates safe gain scheduling as an MDP whose action space is already safety certified. Rather than learning thrust
and torque commands directly, the proposed DQN policy selects gain vectors
from a finite library of pre-certified stabilizing controllers.
Because each admissible action preserves forward invariance of the prescribed
safe set, every policy explored during training
and deployment remains safety-preserving by construction. This distinguishes
the proposed method from prior safe-RL approaches that rely on runtime
correction or auxiliary safety filters, and it yields a more interpretable,
verification-friendly framework for adaptive quadcopter control.

\subsection{Contributions}

We consider the problem of safe learning for dynamical systems with uncertain models through control-theoretic invariance. By modeling the system evolution as a Markov decision process (MDP), safety is embedded directly into the decision-making structure rather than enforced through auxiliary mechanisms. In contrast to existing approaches, the proposed framework establishes safety at the level of the action space, yielding the following contributions:

\noindent \textbf{Contribution 1: Forward-Invariant Action Design for Safe Learning.}
We construct an invariance-induced action space in which each discrete action corresponds to a stabilizing feedback law that guarantees forward invariance of a prescribed admissible set. This yields a finite family of controllers under which the closed-loop system remains safe for all time, independently of the selected action sequence. Consequently, safety is decoupled from the learning process and holds uniformly over all admissible policies, including those encountered during exploration.

\noindent \textbf{Contribution 2: Safe Policy Optimization over Invariance-Certified Action Spaces.}
We show that, under this construction, reinforcement learning can be performed over a safety-certified policy class without the need for runtime constraint handling, action projection, or shielding mechanisms. In contrast to existing safe RL approaches that rely on online optimization or corrective filtering, the proposed framework ensures that every policy generated by the learning algorithm is safe by design. This leads to a simplified learning architecture with reduced computational overhead and improved scalability.

\noindent \textbf{Contribution 3: Safety-Certified Gain Scheduling for Quadcopter Control.}
We validate the proposed framework on a quadcopter hover control problem, where the learning agent performs gain scheduling over a predefined family of stabilizing controllers. The case study demonstrates that adaptive performance optimization can be achieved while preserving strict safety guarantees throughout both training and deployment.

\subsection{Outline}
This paper is organized as follows. Section \ref{Problem Statement} formulates safe learning as an MDP with invariance-certified actions that guarantee safety and reduce learning to optimal control over a safe policy class. Section \ref{Case Study: Quadcopter Hover Regulation} presents a quadcopter hover case study. Section \ref{DQN-based Safe Learning} introduces the DQN-based safe learning framework. Simulation results validating the proposed approach are provided in Section \ref{subsec:sim_setup}, followed by concluding remarks in Section \ref{Conclusion}.
\section{Problem Statement}\label{Problem Statement}

We consider safe learning for a discrete-time dynamical system in which safety is guaranteed by design, independently of the learning process. To this end, we formulate a Markov decision process (MDP)
\[
\mathcal{M}=(\mathcal{X},\mathcal{A},F,r,\gamma),
\]
where $\mathcal{X}\subseteq\mathbb{R}^n$ is a continuous state space, $\mathcal{A}$ is a finite action set, $r:\mathcal{X}\times\mathcal{A}\to\mathbb{R}$ is the stage reward, and $\gamma\in(0,1)$ is the discount factor. The system evolves according to
\[
x_{k+1}=F(x_k,a_k),
\]
where $F:\mathcal{X}\times\mathcal{A}\to\mathcal{X}$ is the closed-loop transition map.

Rather than enforcing safety through online constraint handling or action filtering, we embed safety directly into the admissible action set.

\noindent\textbf{Problem 1 (Invariance-Induced Action Space Design).}
Construct a feedback parameterization and a finite action set $\mathcal{A}$ such that
\[
F(x,a)\in\mathcal{X}, \qquad \forall (x,a)\in\mathcal{X}\times\mathcal{A}.
\]
Then $\mathcal{X}$ is controlled invariant under all admissible actions, and
\[
x_0\in\mathcal{X}\quad\Longrightarrow\quad x_k\in\mathcal{X},\qquad \forall k\ge 0.
\]

\noindent\textbf{Problem 2 (Learning over an Invariant Policy Class).}
Given the action set $\mathcal{A}$ from Problem 1, determine a policy
\[
\pi:\mathcal{X}\to\mathcal{A}
\]
that maximizes
\[
J_\pi(x_0)
=
\mathbb{E}\!\left[\sum_{k=0}^{\infty}\gamma^k r\bigl(x_k,\pi(x_k)\bigr)\right].
\]
Since every action in $\mathcal{A}$ preserves invariance, learning reduces to optimal decision-making over a safety-certified policy class.

As a case study, we consider gain-scheduled quadcopter hover regulation, where each action corresponds to a stabilizing feedback gain selected from an invariance-certified family.

\section{Case Study: Quadcopter Hover Achievement}
\label{Case Study: Quadcopter Hover Regulation}

We instantiate the proposed framework on quadcopter hover regulation. The objective is to drive the vehicle to a desired hover equilibrium while preserving safety. We proceed in two steps: first, we construct an invariance-certified action set; second, we derive a discrete-time transition model for learning.

\subsection{Forward-Invariance-Preserving Control Design}

We consider the control-affine quadcopter dynamics
\begin{equation}
\dot{\mathbf{x}}=\mathbf{f}_0(\mathbf{x})+\mathbf{G}_0\mathbf{u},
\end{equation}
where
\[
\mathbf{x}
=
\begin{bmatrix}
\mathbf{r}^\top &
\mathbf{v}^\top &
\boldsymbol{\eta}^\top &
\dot{\boldsymbol{\eta}}^\top &
T &
\dot T
\end{bmatrix}^\top
\in\mathbb{R}^{14},
\]
\[
\mathbf{u}
=
\begin{bmatrix}
u_T & u_\phi & u_\theta & u_\psi
\end{bmatrix}^\top
\in\mathbb{R}^{4}.
\]
Here, $\mathbf{r}\in\mathbb{R}^3$ and $\mathbf{v}\in\mathbb{R}^3$ denote position and velocity, $\boldsymbol{\eta}=[\phi,\theta,\psi]^\top$ is the Euler-angle vector, and $T$ and $\dot T$ are the thrust deviation and thrust rate, respectively.

The drift vector field and input matrix are
\begin{equation}
\mathbf{f}_0(\mathbf{x})=
\begin{bmatrix}
\mathbf{v} \\[3pt]
-\,g\hat{\mathbf e}_3+\dfrac{mg+T}{m}\mathbf{R}(\boldsymbol{\eta})\hat{\mathbf e}_3\\[6pt]
\dot{\boldsymbol{\eta}} \\[3pt]
\mathbf{0}_{3} \\[3pt]
\dot T \\[3pt]
0
\end{bmatrix},
~
\mathbf{G}_0=
\begin{bmatrix}
\mathbf{0}_{9\times 1} & \mathbf{0}_{9\times 3}\\
\mathbf{0}_{3\times 1} & \mathbf{I}_3\\
0 & \mathbf{0}_{1\times 3}\\
1 & \mathbf{0}_{1\times 3}
\end{bmatrix}.
\end{equation}
Here, $g>0$ is the gravitational constant, $m>0$ is the vehicle mass, $\hat{\mathbf e}_3=[0~0~1]^\top$, and $\mathbf{R}(\boldsymbol{\eta})\in SO(3)$ is the rotation matrix from body to inertial coordinates.

Let $\mathbf r_I^*\in\mathbb R^3$ denote the desired hover position, and let $\mathbf r_d(t)$ be a smooth reference trajectory from $\mathbf r(0)$ to $\mathbf r_I^*$ with bounded derivatives up to fourth order. Define the tracking errors
\[
\mathbf e_r=\mathbf r-\mathbf r_d,~
\mathbf e_v=\mathbf v-\dot{\mathbf r}_d,~
\mathbf e_a=\mathbf a-\ddot{\mathbf r}_d,~
\]
and the external tracking-error state
\begin{equation}
\mathbf z=
\begin{bmatrix}
\mathbf e_r^\top &
\mathbf e_v^\top &
\mathbf e_a^\top &
\mathbf e_j^\top &
\psi &
\dot\psi
\end{bmatrix}^\top
\in\mathbb R^{14}.
\end{equation}

Using differential flatness and dynamic inversion, the external error dynamics can be written as
\begin{equation}
\dot{\mathbf z}
=
\mathbf A_{\mathrm{EXT}}\mathbf z
+
\mathbf B_{\mathrm{EXT}}\mathbf s
-
\mathbf B_{\mathrm{EXT}}\mathbf r_d^{(4)}(t),
\end{equation}
where $\mathbf s\in\mathbb R^4$ is the virtual input,
\[
\mathbf A_{\mathrm{EXT}}=
\begin{bmatrix}
\mathbf{0}_{3} & \mathbf{I}_{3} & \mathbf{0}_{3} & \mathbf{0}_{3} & \mathbf{0} & \mathbf{0}\\
\mathbf{0}_{3} & \mathbf{0}_{3} & \mathbf{I}_{3} & \mathbf{0}_{3} & \mathbf{0} & \mathbf{0}\\
\mathbf{0}_{3} & \mathbf{0}_{3} & \mathbf{0}_{3} & \mathbf{I}_{3} & \mathbf{0} & \mathbf{0}\\
\mathbf{0}_{3} & \mathbf{0}_{3} & \mathbf{0}_{3} & \mathbf{0}_{3} & \mathbf{0} & \mathbf{0}\\
\mathbf{0} & \mathbf{0} & \mathbf{0} & \mathbf{0} & 0 & 1\\
\mathbf{0} & \mathbf{0} & \mathbf{0} & \mathbf{0} & 0 & 0
\end{bmatrix},
~
\mathbf B_{\mathrm{EXT}}=
\begin{bmatrix}
\mathbf{0} & 0\\
\mathbf{0} & 0\\
\mathbf{0} & 0\\
\mathbf{I}_3 & 0\\
\mathbf{0} & 0\\
\mathbf{0} & 1
\end{bmatrix}.
\]
The first three inputs of $\mathbf s$ act on the translational snap dynamics, and the fourth input acts on the yaw acceleration dynamics.

\noindent\textbf{Stability and Forward Invariance:}
We consider a family of feedback laws parameterized by
\[
\mathbf k\in\mathcal K\subset\mathbb R^{14},
\qquad
k_i\in[k_{i,\min},k_{i,\max}],\quad i=1,\dots,14,
\]
and choose
\begin{equation}
\mathbf s=-\mathbf K\mathbf z,
\end{equation}
where $\mathbf K\in\mathbb R^{4\times 14}$ is constructed from $\mathbf k$. The resulting closed-loop dynamics are
\begin{equation}
\dot{\mathbf z}
=
\mathbf A_{\mathrm{cl}}(\mathbf k)\mathbf z
-
\mathbf B_{\mathrm{EXT}}\mathbf r_d^{(4)}(t),
\qquad
\mathbf A_{\mathrm{cl}}(\mathbf k)
\triangleq
\mathbf A_{\mathrm{EXT}}-\mathbf B_{\mathrm{EXT}}\mathbf K.
\end{equation}

\begin{theorem}
Consider the closed-loop system
\begin{equation}
\dot{\mathbf z}
=
\mathbf A_{\mathrm{cl}}(\mathbf k)\mathbf z
-
\mathbf B_{\mathrm{EXT}}\mathbf r_d^{(4)}(t),
\end{equation}
where
\[
\mathbf A_{\mathrm{cl}}(\mathbf k)
=
\mathbf A_{\mathrm{EXT}}-\mathbf B_{\mathrm{EXT}}\mathbf K.
\]
Assume that $\mathbf A_{\mathrm{cl}}(\mathbf k)$ is Hurwitz for all $\mathbf k\in\mathcal K$, and that
\[
\|\mathbf r_d^{(4)}(t)\|\le \bar r_4,
\qquad \forall t\ge 0,
\]
for some constant $\bar r_4>0$. Then, for each $\mathbf k\in\mathcal K$, the closed-loop error dynamics are input-to-state stable with respect to the input $\mathbf r_d^{(4)}(t)$. Moreover, for each $\mathbf k\in\mathcal K$, there exists a symmetric positive definite matrix $\mathbf P(\mathbf k)$ such that the quadratic Lyapunov function
\begin{equation}
V(\mathbf z)=\mathbf z^\top \mathbf P(\mathbf k)\mathbf z
\end{equation}
satisfies
\begin{equation}
\dot V(\mathbf z)
\le
-\alpha \|\mathbf z\|^2+\beta \bar r_4^2,
\end{equation}
for some constants $\alpha,\beta>0$.
\end{theorem}

\begin{proof}
Fix any $\mathbf k\in\mathcal K$. Since $\mathbf A_{\mathrm{cl}}(\mathbf k)$ is Hurwitz, for any symmetric positive definite matrix $\mathbf Q$ there exists a unique symmetric positive definite matrix $\mathbf P(\mathbf k)$ satisfying the Lyapunov equation
\begin{equation}
\mathbf A_{\mathrm{cl}}(\mathbf k)^\top \mathbf P(\mathbf k)
+
\mathbf P(\mathbf k)\mathbf A_{\mathrm{cl}}(\mathbf k)
=
-\mathbf Q.
\end{equation}
By choosing $\mathbf Q=\mathbf I$, the function
\[
V(\mathbf z)=\mathbf z^\top \mathbf P(\mathbf k)\mathbf z
\]
is positive definite and radially unbounded. Differentiating $V$ along the trajectories of the closed-loop system gives
\begin{align}
\dot V(\mathbf z)
&=
\mathbf z^\top
\bigl(
\mathbf A_{\mathrm{cl}}(\mathbf k)^\top \mathbf P(\mathbf k)
+
\mathbf P(\mathbf k)\mathbf A_{\mathrm{cl}}(\mathbf k)
\bigr)\mathbf z\\
-&
2\mathbf z^\top \mathbf P(\mathbf k)\mathbf B_{\mathrm{EXT}}\mathbf r_d^{(4)}(t) =
-\mathbf z^\top \mathbf Q \mathbf z
-
2\mathbf z^\top \mathbf P(\mathbf k)\mathbf B_{\mathrm{EXT}}\mathbf r_d^{(4)}.
\end{align}
By provoking the Cauchy--Schwarz inequality, we obtain
\begin{align}
\dot V(\mathbf z)
&\le
-\|\mathbf z\|^2
+
2\|\mathbf P(\mathbf k)\mathbf B_{\mathrm{EXT}}\|\,\|\mathbf z\|\,\|\mathbf r_d^{(4)}(t)\| \\
&\le
-\|\mathbf z\|^2
+
2\|\mathbf P(\mathbf k)\mathbf B_{\mathrm{EXT}}\|\,\bar r_4\,\|\mathbf z\|.
\end{align}
By applying Young's inequality,
\[
2ab\le \varepsilon a^2+\frac{1}{\varepsilon}b^2,
\qquad \forall\,\varepsilon>0,
\]
with
\[
a=\|\mathbf z\|,
\qquad
b=\|\mathbf P(\mathbf k)\mathbf B_{\mathrm{EXT}}\|\bar r_4,
\]
we obtain
\begin{equation}
\dot V(\mathbf z)
\le
-(1-\varepsilon)\|\mathbf z\|^2
+
\frac{\|\mathbf P(\mathbf k)\mathbf B_{\mathrm{EXT}}\|^2}{\varepsilon}\bar r_4^2.
\end{equation}
Finally, by choosing any $\varepsilon\in(0,1)$ and defining
\[
\alpha=1-\varepsilon,
\qquad
\beta=\frac{\|\mathbf P(\mathbf k)\mathbf B_{\mathrm{EXT}}\|^2}{\varepsilon},
\]
we obtain
\begin{equation}
\dot V(\mathbf z)
\le
-\alpha \|\mathbf z\|^2
+
\beta \|\mathbf r_d^{(4)}(t)\|^2.
\end{equation}
This inequality implies that $\dot V(\mathbf z) < 0$ whenever
\[
\|\mathbf z\|^2 > \frac{\beta}{\alpha}\|\mathbf r_d^{(4)}(t)\|^2.
\]
Hence, the trajectories enter and remain in the compact set
\[
\Omega =
\left\{
\mathbf z :
\|\mathbf z\|^2 \le \frac{\beta}{\alpha}\bar r_4^2
\right\},
\]
which establishes uniform ultimate boundedness.

Moreover, since $V$ is positive definite and radially unbounded and satisfies the dissipation inequality
\[
\dot V \le -\alpha(\|\mathbf z\|) + \gamma(\|\mathbf r_d^{(4)}(t)\|),
\]
the closed-loop system is input-to-state stable with respect to $\mathbf r_d^{(4)}(t)$.
\end{proof}

For each admissible gain choice, define
\[
\Omega_\rho
=
\{\mathbf z\in\mathbb R^{14}:V(\mathbf z)\le \rho\}.
\]
If $\rho>\beta \bar r_4^2/\alpha$, then $\Omega_\rho$ is forward invariant for the closed-loop error dynamics. Consequently, the state-space safe set
\[
\mathcal X
=
\{x: \mathbf z(x)\in\Omega_\rho\}
\]
is forward invariant under every admissible action $a\in\mathcal A$ induced by the gain family $\mathcal K$.


\noindent\textbf{Nonlinear closed-loop dynamics:}
Through dynamic inversion, the external input $\mathbf{s}$ is mapped to the physical control input $\mathbf{u}$ via
\[
\mathbf{s} = \mathbf{M}(\mathbf{x})\mathbf{u} + \mathbf{n}(\mathbf{x}),
\]
so that
\[
\mathbf{u} = \mathbf{M}^{-1}(\mathbf{x})\bigl(\mathbf{s} - \mathbf{n}(\mathbf{x})\bigr).
\]
Substituting this into the control-affine dynamics yields
\begin{equation}\label{dyn8}
\dot{\mathbf{x}} = \mathbf{f}(\mathbf{x}) + \mathbf{G}(\mathbf{x})\mathbf{k},
\end{equation}
which defines the nonlinear closed-loop system parameterized by $\mathbf{k}$.
Define the admissible set
\[
\mathcal{X}=\{x:\|\mathbf{z}(x)\|\le \delta\}.
\]
If $\mathcal{K}$ is chosen such that the error dynamics are uniformly asymptotically stable for all $\mathbf{k}\in\mathcal{K}$, then $\mathcal{X}$ is forward invariant under all admissible actions:
\[
x\in\mathcal{X}\quad\Longrightarrow\quad F(x,a)\in\mathcal{X},\qquad \forall a\in\mathcal{A},
\]
where
\[
\mathcal{A}=\{\mathbf{k}^{(1)},\dots,\mathbf{k}^{(N)}\}\subset\mathcal{K}.
\]

\subsection{Control-Oriented Discrete-Time Dynamics}

Under a zero-order hold discretization, the quadcopter dynamics \eqref{dyn8} are expressed as
\begin{equation}
x_{k+1} = F(x_k, a_k),
\end{equation}
where $F$, computed numerically (e.g., via a Runge--Kutta scheme), defines the MDP transition map, and $a_k \in \mathcal{A}$ denotes the control gain vector applied at discrete time $k$.

\section{DQN-based Safe Learning}
\label{DQN-based Safe Learning}

Given $\mathcal{A}$, policy optimization is formulated as a discrete-action reinforcement learning problem. For this problem, $x_k\in \mathbb{R}^{14}$
is considered as the state, and the reward is given by
\begin{equation}
\begin{split}
r(x_k,a_k)
=&-
\Big(
w_r \|\mathbf{e}_r\|^2
+ w_v \|\mathbf{e}_v\|^2
+ w_\eta \|\boldsymbol{\eta}\|^2
+ w_\omega \|\boldsymbol{\omega}\|^2
\Big)\\
-&\,w_u \|\mathbf{u}\|^2
- w_s \mathbf{1}\{a_k\neq a_{k-1}\},
\end{split}
\end{equation}
where \(\mathbf{e}_r=\mathbf{r}-\mathbf{r}_d\) and \(\mathbf{e}_v=\mathbf{v}-\dot{\mathbf{r}}_d\) denote the position- and velocity-tracking errors, \(\boldsymbol{\eta}=[\phi,\theta,\psi]^\top\) is the attitude vector, and \(\boldsymbol{\omega}\) is the body angular velocity. The term \(w_r\|\mathbf{e}_r\|^2\) promotes accurate convergence to the desired hover position, while \(w_v\|\mathbf{e}_v\|^2\) suppresses residual translational motion and improves transient damping. The attitude penalty \(w_\eta\|\boldsymbol{\eta}\|^2\) discourages excessive roll, pitch, and yaw excursions, and the angular-rate penalty \(w_\omega\|\boldsymbol{\omega}\|^2\) reduces aggressive rotational motion and oscillatory behavior. The control-effort term \(w_u\|\mathbf{u}\|^2\) regularizes actuator usage and prevents unnecessarily large control inputs. Finally, the switching penalty \(w_s \mathbf{1}\{a_k\neq a_{k-1}\}\) discourages frequent changes between gain selections, thereby promoting smoother controller switching and reducing chattering in the learned policy. Large negative terminal penalties are assigned when admissibility conditions are violated, so that unsafe or diverging trajectories become strongly suboptimal.

The action-value function is approximated by a neural network
\[
Q(x,a;\theta),
\]
where \(\theta\) denotes the trainable parameters of the Q-network. The network is trained using standard DQN with experience replay and a target network. At decision time \(k\), the action is selected according to an \(\epsilon\)-greedy policy:
\[
a_k=
\begin{cases}
\text{a random action in } \mathcal{A}, & \text{with probability } \epsilon,\\[2mm]
\displaystyle \arg\max_{a\in\mathcal{A}} Q(x_k,a;\theta), & \text{with probability } 1-\epsilon.
\end{cases}
\]
The loss is
\begin{subequations}
    \begin{equation}
\mathcal{L}(\theta)=
\mathbb{E}\!\left[
\bigl(Q(x_k,a_k;\theta)-y_k\bigr)^2
\right],
\end{equation}
\begin{equation}
y_k=r_k+\gamma\max_{a'}Q(x_{k+1},a';\theta^-).
\end{equation}

\end{subequations}

Since all actions preserve invariance, every policy encountered during training is safety-preserving, and exploration requires no constraint handling.

\section{Simulation Results}
\label{subsec:sim_setup}

We evaluate the proposed DQN-based gain-scheduling controller on the quadcopter
case study using a high-fidelity nonlinear simulation with Euler ZYX attitude
representation and thrust/torque actuation. The physical state is
$
\mathbf{x} =
[\mathbf{r}^\top,\mathbf{v}^\top,\boldsymbol{\eta}^\top,\boldsymbol{\omega}^\top,T_{\mathrm{dev}},\dot T]^\top
\in \mathbb{R}^{14},
$
where $\mathbf{r},\mathbf{v}\in\mathbb{R}^3$ denote inertial position and velocity,
$\boldsymbol{\eta}=[\phi,\theta,\psi]^\top$ denotes the Euler-angle attitude,
$\boldsymbol{\omega}\in\mathbb{R}^3$ denotes the body angular velocity,
and $T_{\mathrm{dev}}$ and $\dot T$ capture thrust-deviation dynamics.
The vehicle parameters are fixed to mass $m=1.5~\mathrm{kg}$,
gravitational acceleration $g=9.81~\mathrm{m/s^2}$, and inertia matrix
$\mathbf{I}=\mathrm{diag}(0.02,\,0.02,\,0.04)~\mathrm{kg\,m^2}$.
The simulation is integrated with time step $\Delta t = 0.01~\mathrm{s}$,
and each episode lasts $10~\mathrm{s}$.

\begin{figure}[!h]
  \centering
  \includegraphics[width=0.95\linewidth]{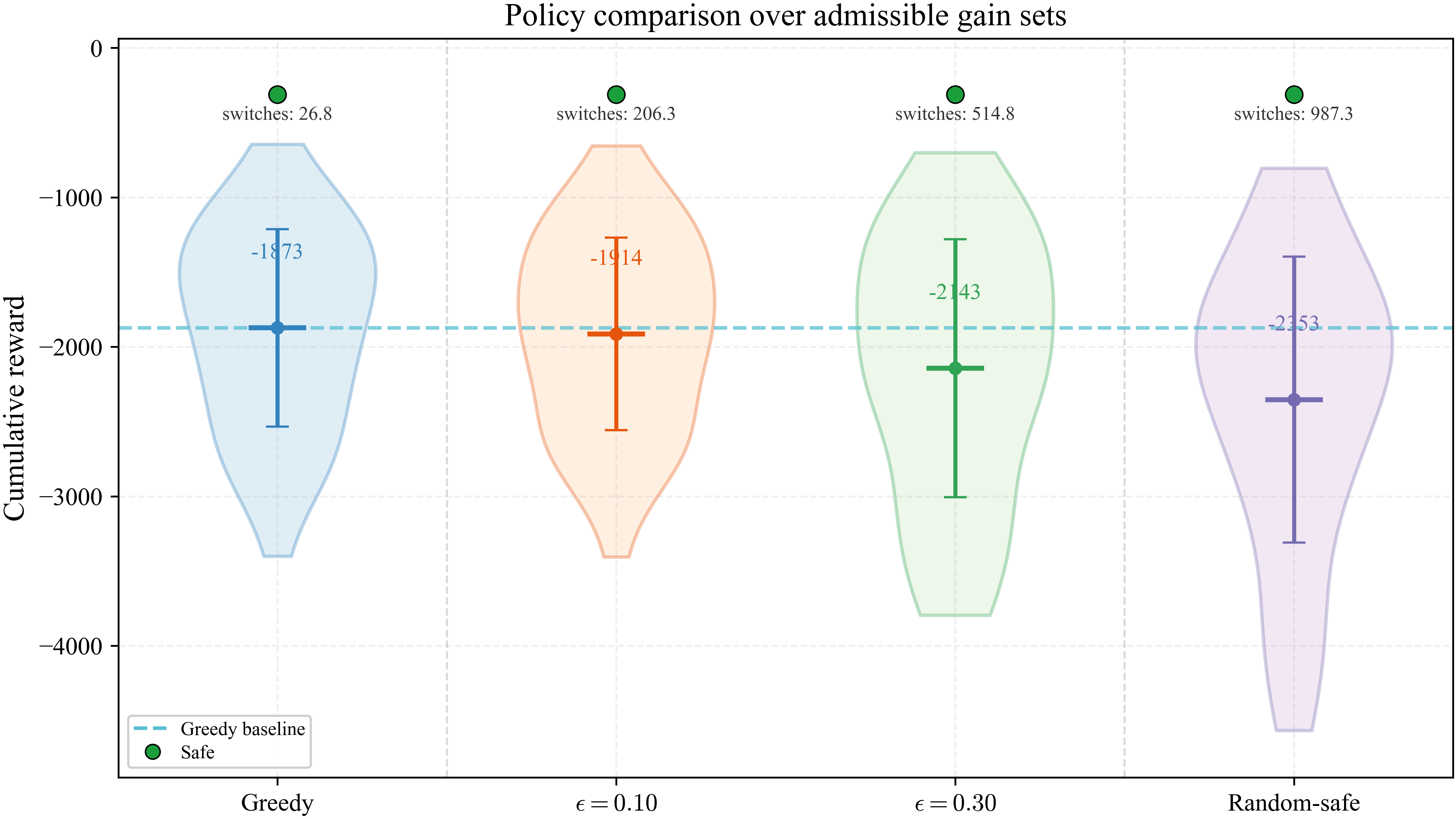}
  \caption{Policy comparison over the admissible gain set. The results show that safety is preserved across all evaluated policies, whereas performance and switching efficiency depend strongly on the policy. Violin plots show the distribution of cumulative reward across 40 evaluation rollouts for each policy. Colored markers and vertical bars denote the sample mean and one standard deviation, respectively, while the dashed horizontal line indicates the greedy-policy mean. The annotated values report the average number of action switches.}
  \label{fig:policy_compare}
\end{figure}

The reference trajectory is generated using a smooth quintic time-scaling over
$t\in[0,T_f]$ with $T_f=5~\mathrm{s}$. For $t>T_f$, the desired position is held
at $\mathbf{r}_d(T_f)$, while the desired velocity, acceleration, jerk, and snap
are set to zero so that the vehicle transitions to a hover condition after the
maneuver.
\begin{figure}[!h]
  \centering
  \includegraphics[width=0.85\linewidth]{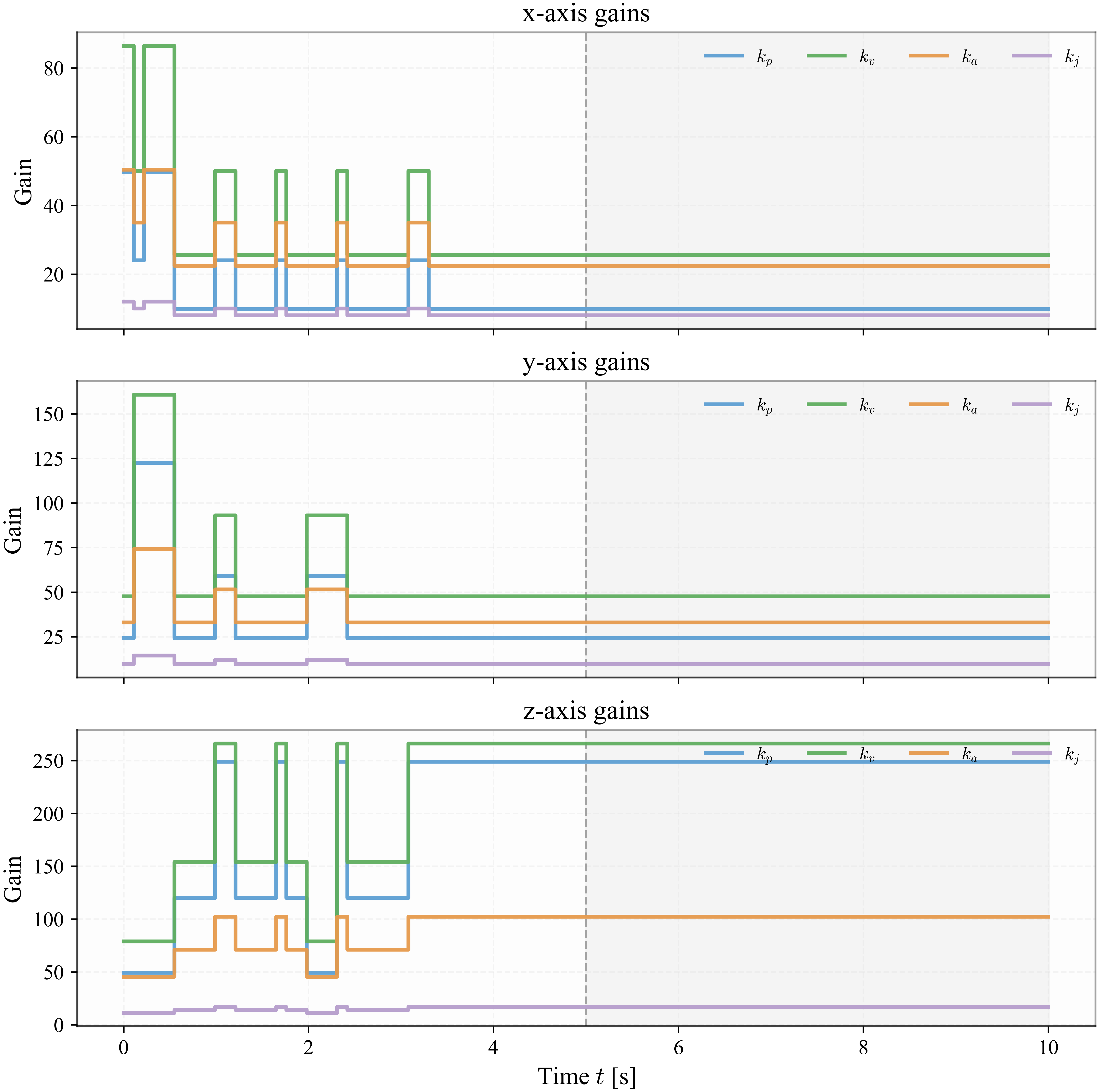}
  \caption{Representative rollout of the translational gains selected by the trained DQN along the $x$-, $y$-, and $z$-axes. The learned scheduler switches more actively during the initial transient and then transitions to a nearly constant gain regime as the quadcopter approaches the terminal hover condition. The $z$-axis generally exhibits larger gain magnitudes than the $x$- and $y$-axes, consistent with the stronger vertical regulation required for altitude control.}
  \label{fig:dqn_gains}
\end{figure}

\begin{table}[!h]
\centering
\caption{Closed-loop attitude safety statistics over 40 evaluation rollouts. All policies completed all 40 rollouts without unsafe termination. Lower values indicate smaller pitch excursions.}
\label{tab:policy_summary}
\setlength{\tabcolsep}{5pt}
\renewcommand{\arraystretch}{1.08}
\begin{tabular}{lcc}
\toprule
Policy & Mean peak $|\theta|$ (rad) & Worst-case $|\theta|$ (rad) \\
\midrule
Greedy           & \textbf{0.170} & \textbf{0.314} \\
$\epsilon=0.10$  & 0.173          & \textbf{0.314} \\
$\epsilon=0.30$  & 0.171          & 0.320 \\
Random-safe      & 0.174          & 0.328 \\
\bottomrule
\end{tabular}
\end{table}

The DQN observation is formed by concatenating the 14-dimensional physical state
with the scalar phase variable $\min(t/T_f,1)$, resulting in a 15-dimensional
input. At each decision step, the agent selects one discrete action from a finite
table of pre-computed stabilizing gain vectors
$
\mathbf{k}=[k_1,\dots,k_{14}]^\top.
$
In the current implementation, translational gains are selected separately along
the $x$-, $y$-, and $z$-axes, whereas the yaw gains are treated independently to
reflect the distinct second-order yaw dynamics. This parameterization enlarges the
admissible discrete action set and enables direction-dependent modulation of
feedback aggressiveness during the maneuver. To avoid excessive switching, a
dwell-time constraint is imposed so that each selected action is held for a fixed
number of time steps before another switch is allowed.

\begin{figure}[!h]
  \centering
  \includegraphics[width=0.85\linewidth]{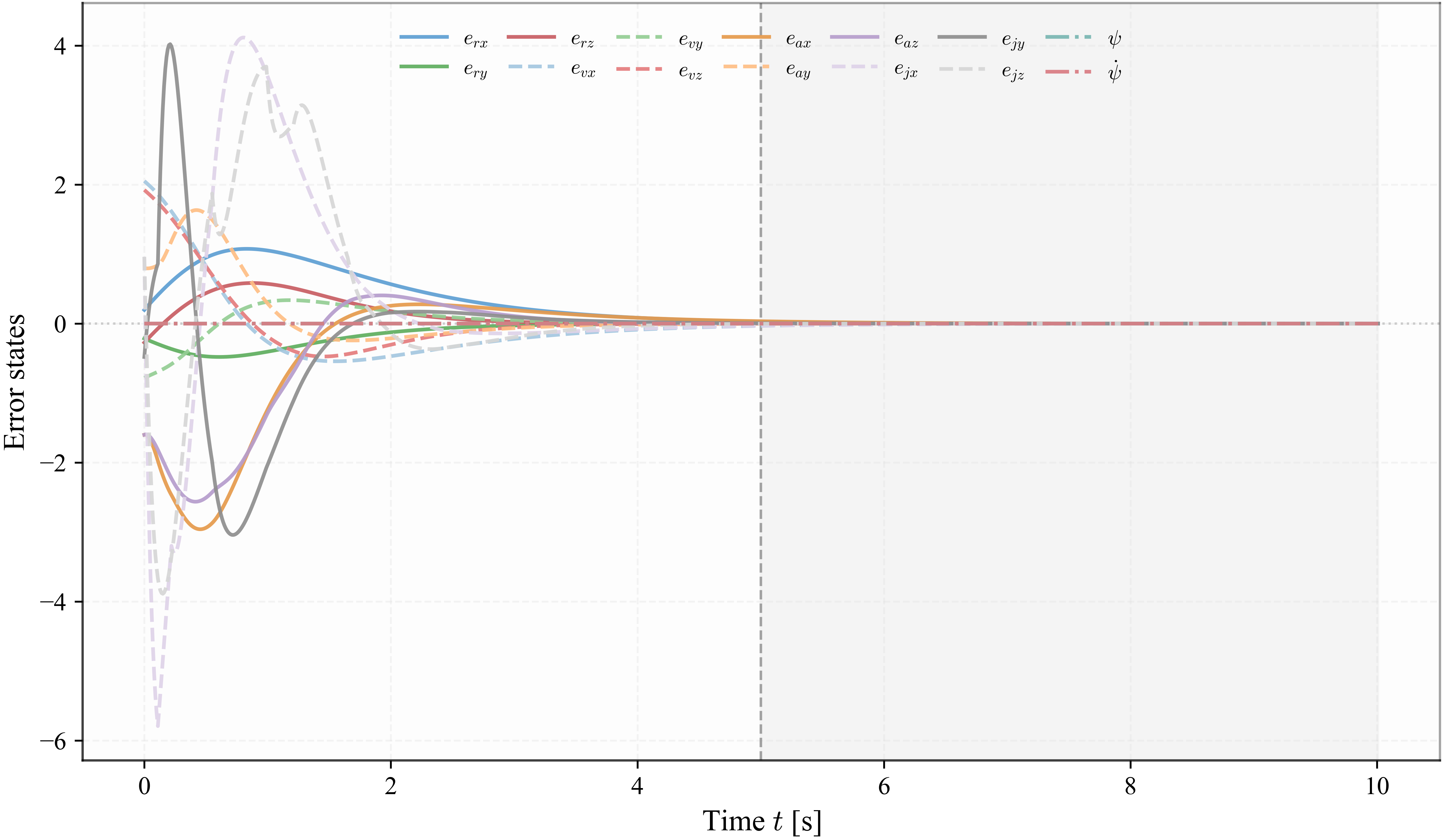}
  \caption{Representative DQN rollout of the external error states. The translational error components at the position, velocity, acceleration, and jerk levels decrease toward small values, while the yaw channel remains bounded, indicating stable closed-loop regulation under the learned gain-scheduling policy.}
  \label{fig:dqn_states}
\end{figure}

We evaluate the proposed framework from two complementary perspectives:
\emph{(i)} whether safety is preserved under different deployment-time
evaluation policies defined over the same admissible gain set, and
\emph{(ii)} whether an offline-trained policy yields feasible closed-loop hover
regulation in a representative rollout.

Our main quantitative result is shown in
Figure~\ref{fig:policy_compare}, which compares four evaluation policies over
the same finite gain table: the deployed greedy DQN policy, two fixed
$\epsilon$-greedy policies with $\epsilon=0.10$ and $\epsilon=0.30$, and a
random-safe policy that samples uniformly from the admissible action set.
Here, the non-greedy policies are included as alternative evaluation policies
over the same certified controller set, rather than as part of the deployment
strategy itself. The violin plots summarize the cumulative reward
distributions across 40 rollouts, and the annotations report the average
number of gain switches. This comparison is particularly informative because it
separates policy-dependent performance from policy-independent safety.

\begin{figure}[!h]
  \centering
  \includegraphics[width=0.85\linewidth]{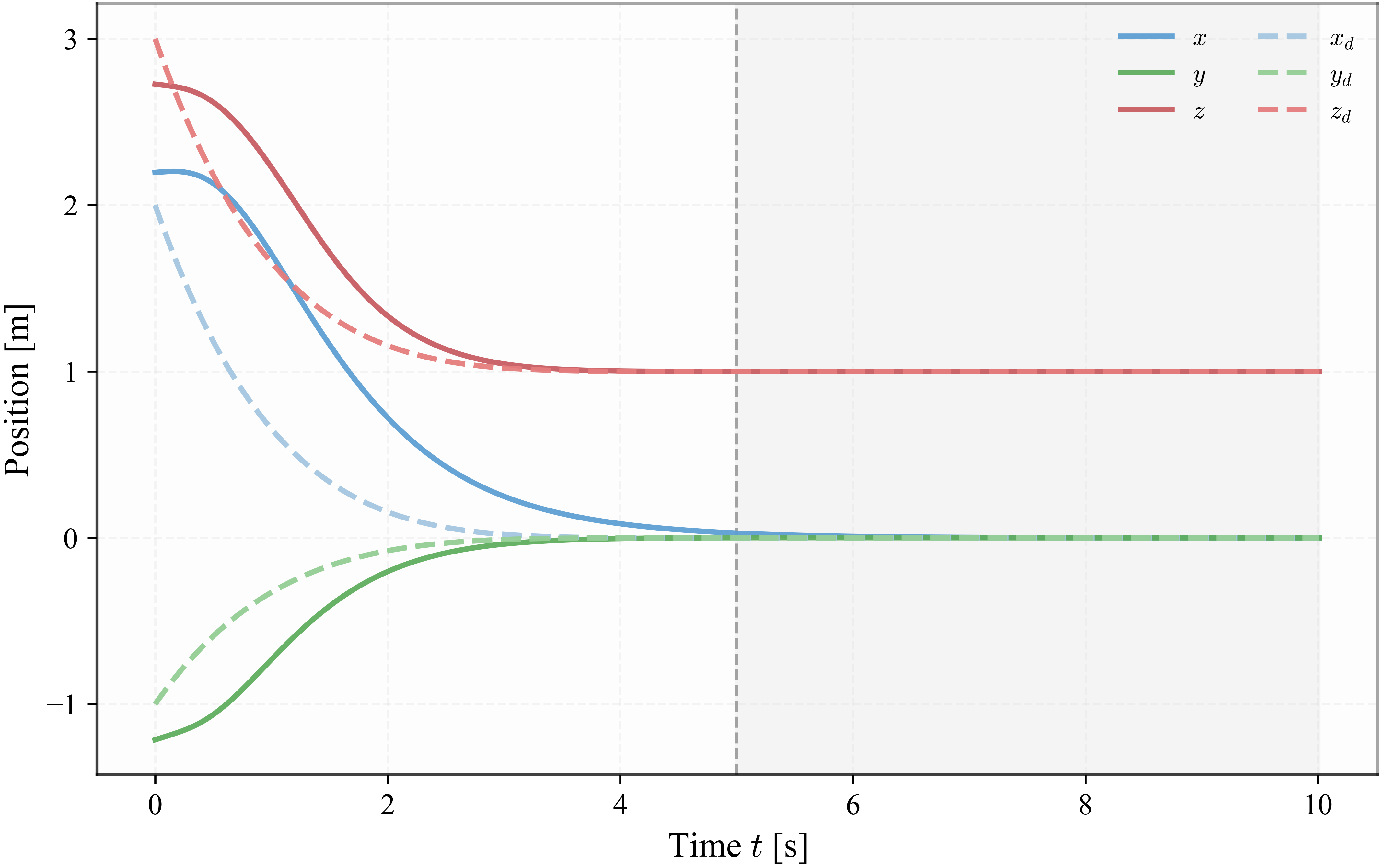}
  \caption{Physical evaluation of inertial position tracking. The quadcopter follows the desired trajectory during the quintic maneuver and settles toward the terminal hover condition after $T_f$, where the reference position is held constant at $\mathbf{r}_d(T_f)$.}
  \label{fig:pos_vs_des}
\end{figure}

\begin{figure}[!h]
  \centering
  \includegraphics[width=0.78\linewidth]{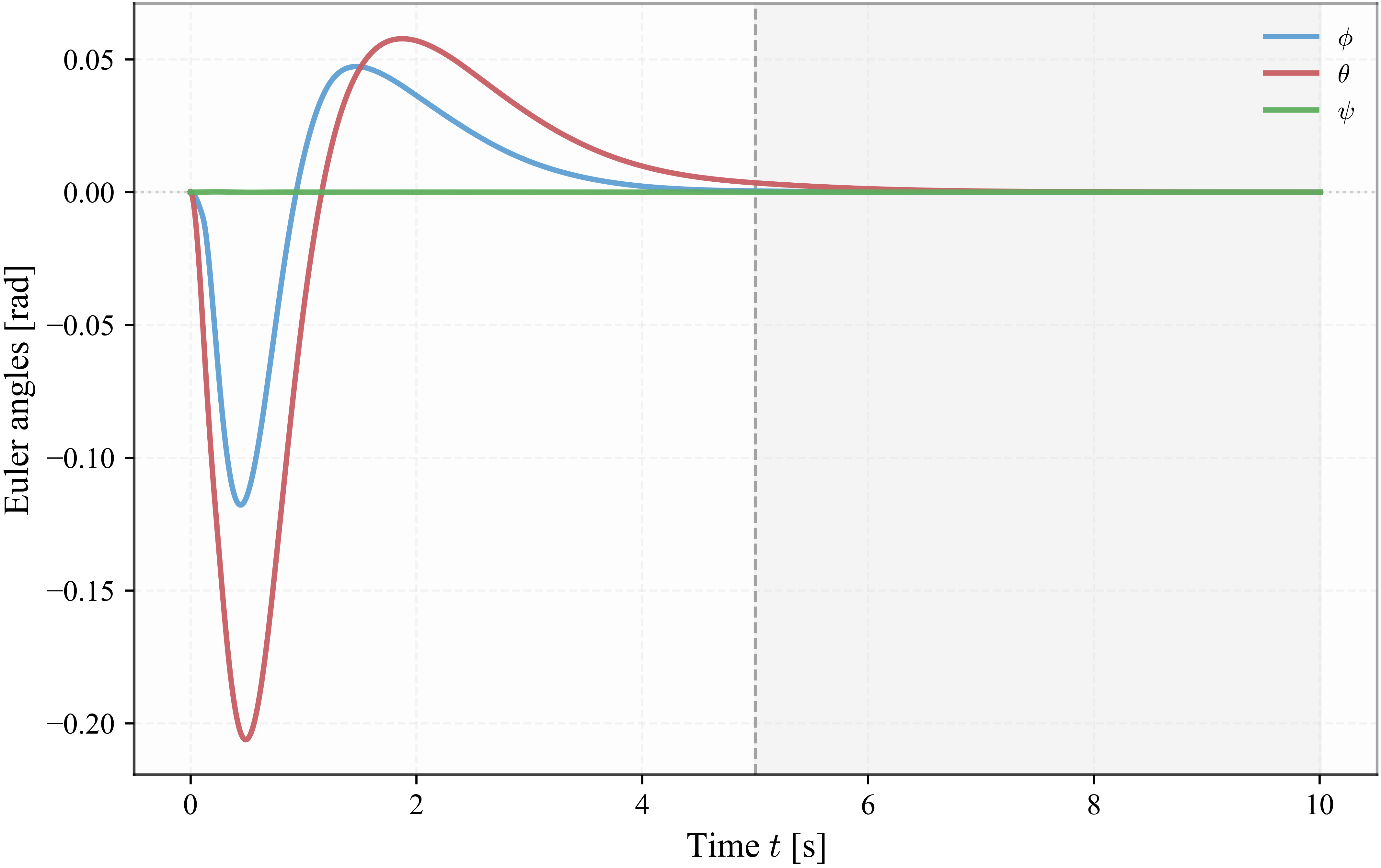}
  \caption{Physical evaluation of the Euler angles $(\phi,\theta,\psi)$. Attitude excursions remain bounded during the maneuver and decay toward small values as the vehicle approaches steady hover.}
  \label{fig:euler}
\end{figure}

\begin{figure}[!h]
  \centering
  \includegraphics[width=0.82\linewidth]{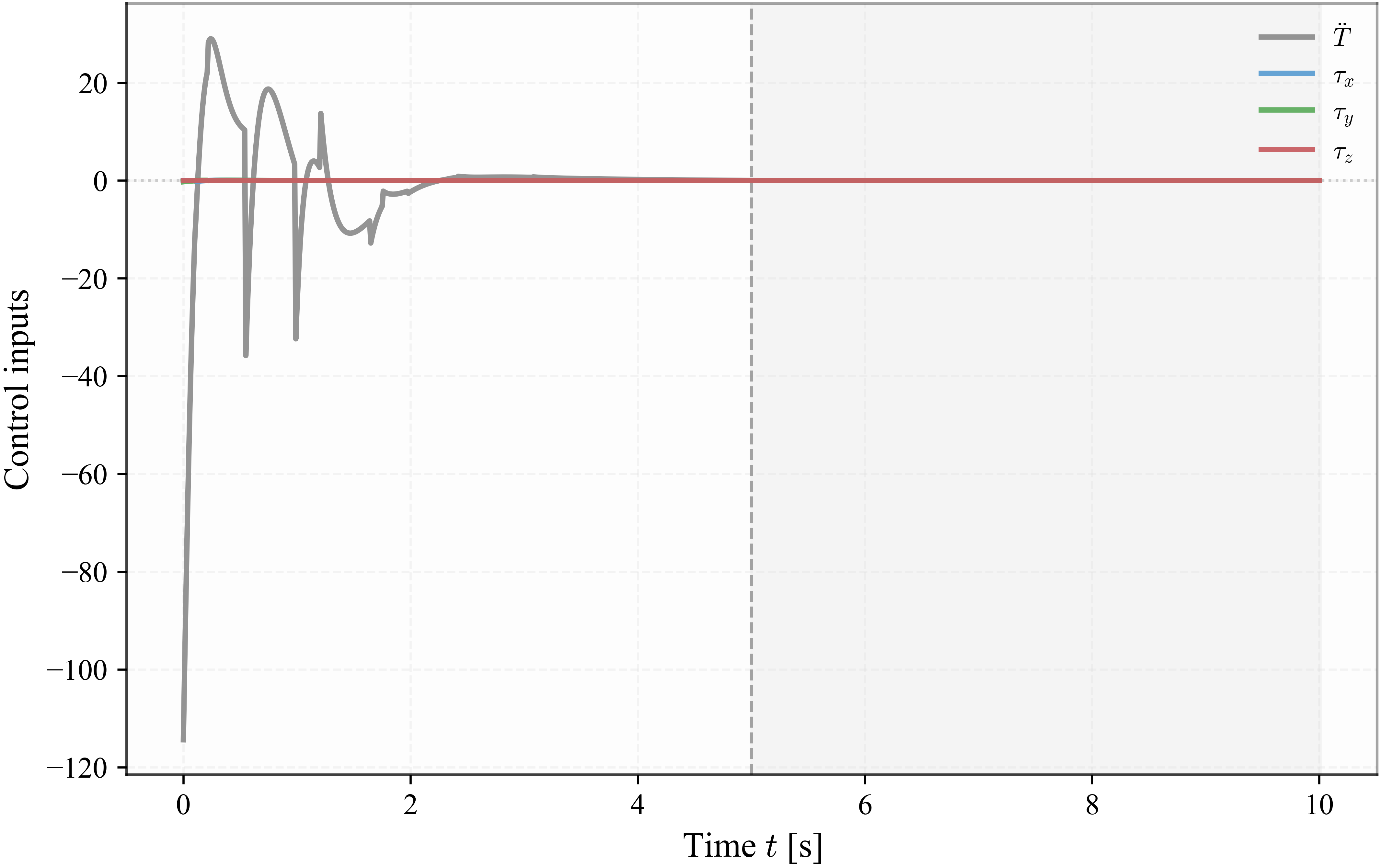}
  \caption{Physical evaluation of the control inputs. The thrust second-derivative command $\ddot T$ and body torques $\boldsymbol{\tau}$ are largest during the initial transient and decrease as the tracking errors are reduced.}
  \label{fig:controls}
\end{figure}

\begin{figure}[!h]
  \centering
  \includegraphics[width=0.82\linewidth]{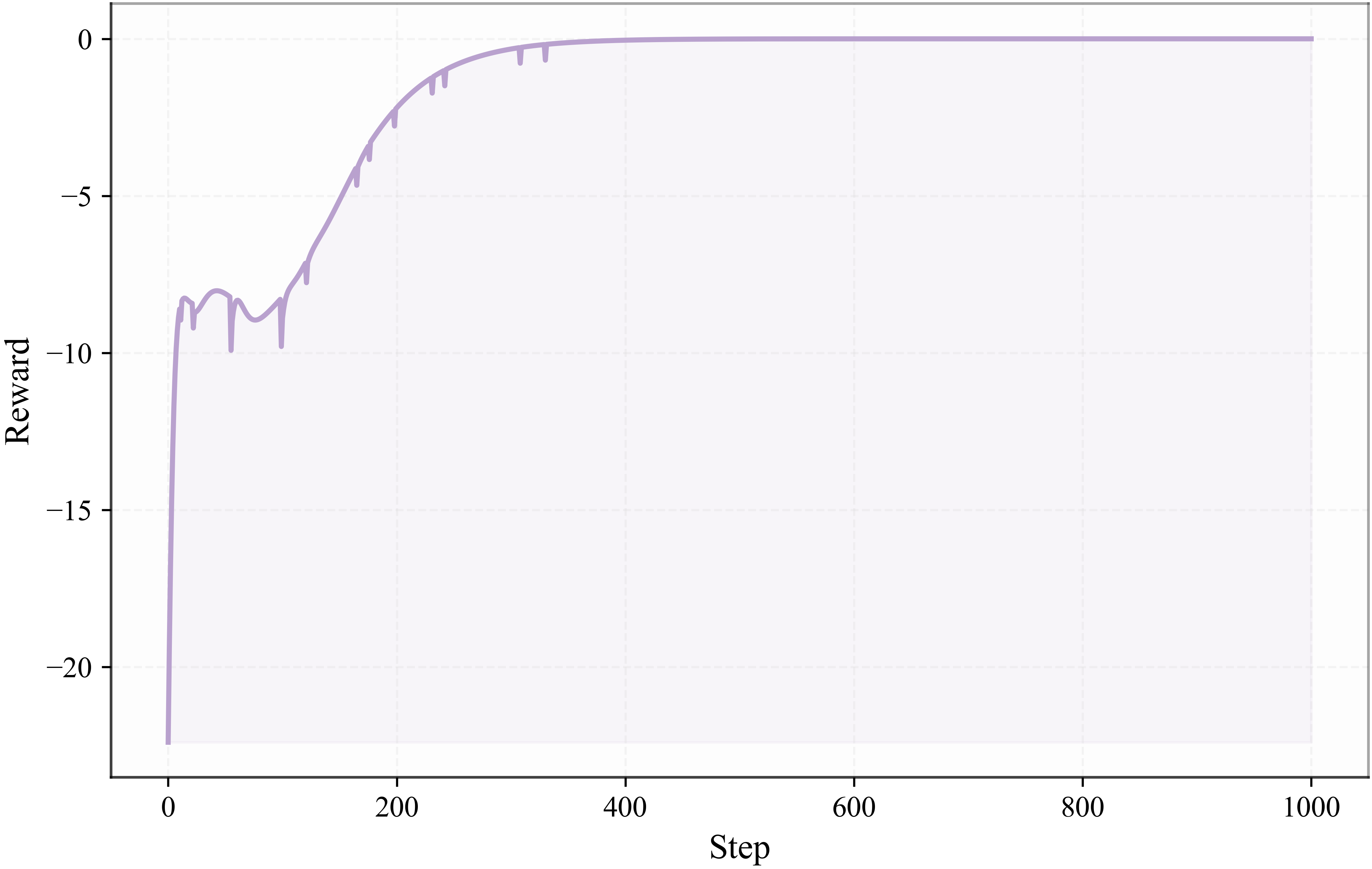}
  \caption{Physical evaluation of the per-step reward. The reward improves over the rollout and approaches zero as the tracking errors and control effort decrease.}
  \label{fig:reward}
\end{figure}

A consistent pattern emerges across all four evaluation policies. No unsafe
termination was observed in any of the 40 rollouts, yielding an empirical
unsafe rate of zero for the greedy, $\epsilon$-greedy, and random-safe
evaluations. Likewise, the violation counts remained zero for all monitored
categories, including attitude, position, velocity, and non-finite state
violations. These results provide strong empirical support for the central
claim of the proposed framework: although the choice of evaluation policy
significantly affects efficiency, all policies remain confined to safe
closed-loop behavior because they operate over the same admissible gain
construction.

Table~\ref{tab:policy_summary} further complements
Figure~\ref{fig:policy_compare} with attitude-safety statistics under the same
40-rollout protocol. Across all policies, the pitch excursions remain tightly
bounded, with mean peak $|\theta|$ values ranging from $0.170$ to $0.174$ rad
and worst-case values ranging from $0.314$ to $0.328$ rad. The greedy policy
achieves the smallest mean peak pitch excursion, while the worst-case pitch
excursion remains similar across all policies. Together,
Figure~\ref{fig:policy_compare} and Table~\ref{tab:policy_summary} reinforce a
clear conclusion: safety is induced primarily by the admissible controller
set, rather than by any particular evaluation policy.

In contrast, the performance metrics vary substantially across policies. The
deployed greedy DQN policy achieves the highest average cumulative reward
($-1872.67$) and the fewest gain switches on average ($26.8$). As the
evaluation policy becomes increasingly exploratory, performance degrades and
switching activity grows sharply: the $\epsilon=0.10$ policy achieves an
average cumulative reward of $-1913.58$ with $206.33$ switches, the
$\epsilon=0.30$ policy yields $-2143.43$ with $514.78$ switches, and the
random-safe policy yields $-2352.84$ with $987.30$ switches. This reveals a
clean separation of roles: policy optimization determines closed-loop
efficiency and switching behavior, whereas safety is inherited from the
admissible gain construction itself.

To illustrate the resulting closed-loop behavior, we next consider a
representative rollout under the deployed greedy DQN policy.
Figure~\ref{fig:dqn_gains} shows the gain-scheduling behavior along the
$x$-, $y$-, and $z$-axes. The policy switches more actively during the initial
transient, then settles into a nearly constant regime after roughly
$3\,\mathrm{s}$ as the quadcopter approaches hover. The selected gains are
also axis-dependent, with the $z$-axis generally requiring larger values than
the $x$- and $y$-axes, consistent with the stronger vertical regulation
demands imposed by gravity and thrust.

The remaining rollout figures provide qualitative evidence that the resulting
closed-loop response is well behaved. In
Figures~\ref{fig:dqn_states}--\ref{fig:euler}, the state, trajectory, and
attitude responses exhibit bounded transients and converge toward the desired
hover condition. The position-related errors decrease after the initial
maneuver, the inertial trajectory tracks the desired reference and settles
near the terminal hover point, and the Euler angles remain bounded
throughout. Figure~\ref{fig:controls} further shows that the thrust and torque
commands are concentrated in the initial transient and decrease as the vehicle
approaches hover, while Figure~\ref{fig:reward} shows that the per-step reward
improves accordingly. Together, these rollout visualizations confirm that the
offline-trained gain schedule yields feasible and stable closed-loop
regulation for the quadcopter case study.

\section{Conclusion}\label{Conclusion}

This paper presented a safe reinforcement learning framework based on
invariance-induced action-space design. By constructing a finite admissible
action set in which each action corresponds to a stabilizing feedback law, the
framework embeds safety directly into the control architecture and preserves
forward invariance of a prescribed safe state set by construction. The
quadcopter hover-control results demonstrated that the proposed formulation
separates two roles that are often intertwined in safe learning: safety is
determined by the admissible controller set, whereas learning improves
closed-loop performance within that set. These results indicate that forward-invariance-certified action design provides
a useful foundation for safe learning in nonlinear systems. Future work will
extend this framework to autonomous driving, where adaptive decision making
must operate over safety-certified steering and braking actions under
lane-keeping, obstacle-avoidance, and vehicle-stability constraints.

\bibliographystyle{IEEEtran}
\bibliography{CDC-Ref}

@inproceedings{Lee2010CDC,
  author    = {Taeyoung Lee and Melvin Leok and N. Harris McClamroch},
  title     = {Geometric Tracking Control of a Quadrotor {UAV} on {SE}(3)},
  booktitle = {49th IEEE Conference on Decision and Control (CDC)},
  year      = {2010},
  pages     = {5420--5425},
  doi       = {10.1109/CDC.2010.5717652}
}

@inproceedings{Mellinger2011ICRA,
  author    = {Daniel Mellinger and Vijay Kumar},
  title     = {Minimum Snap Trajectory Generation and Control for Quadrotors},
  booktitle = {2011 IEEE International Conference on Robotics and Automation (ICRA)},
  year      = {2011},
  pages     = {2520--2525},
  doi       = {10.1109/ICRA.2011.5980409}
}

@inproceedings{Milhim2010AIAA,
  author    = {Alaeddin Milhim and Youmin Zhang and Camille{-}Alain Rabbath},
  title     = {Gain Scheduling Based {PID} Controller for Fault Tolerant Control of Quad-Rotor {UAV}},
  booktitle = {AIAA Infotech@Aerospace Conference},
  year      = {2010},
  address   = {Atlanta, GA, USA},
  doi       = {10.2514/6.2010-3530}
}

@article{Blanchini1999Automatica,
  author    = {Franco Blanchini},
  title     = {Set Invariance in Control},
  journal   = {Automatica},
  volume    = {35},
  number    = {11},
  pages     = {1747--1767},
  year      = {1999},
  doi       = {10.1016/S0005-1098(99)00113-2}
}

@article{Ames2017TAC,
  author    = {Aaron D. Ames and Xiangru Xu and Jessy W. Grizzle and Paulo Tabuada},
  title     = {Control Barrier Function Based Quadratic Programs for Safety-Critical Systems},
  journal   = {IEEE Transactions on Automatic Control},
  volume    = {62},
  number    = {8},
  pages     = {3861--3876},
  year      = {2017},
  doi       = {10.1109/TAC.2016.2638961}
}

@article{Koch2019TCPS,
  author       = {William Koch and Renato Mancuso and Richard West and Azer Bestavros},
  title        = {Reinforcement Learning for {UAV} Attitude Control},
  journal      = {ACM Transactions on Cyber-Physical Systems},
  volume       = {3},
  number       = {2},
  articleno    = {22},
  pages        = {22:1--22:21},
  year         = {2019},
  doi          = {10.1145/3301273}
}

@article{Mnih2015Nature,
  author    = {Volodymyr Mnih and Koray Kavukcuoglu and David Silver and Andrei A. Rusu and Joel Veness and Marc G. Bellemare and Alex Graves and Martin Riedmiller and Andreas K. Fidjeland and Georg Ostrovski and Stig Petersen and Charles Beattie and Amir Sadik and Ioannis Antonoglou and Helen King and Dharshan Kumaran and Daan Wierstra and Shane Legg and Demis Hassabis},
  title     = {Human-Level Control through Deep Reinforcement Learning},
  journal   = {Nature},
  volume    = {518},
  number    = {7540},
  pages     = {529--533},
  year      = {2015},
  doi       = {10.1038/nature14236}
}

@inproceedings{Achiam2017ICML,
  author    = {Joshua Achiam and David Held and Aviv Tamar and Pieter Abbeel},
  title     = {Constrained Policy Optimization},
  booktitle = {Proceedings of the 34th International Conference on Machine Learning (ICML)},
  series    = {Proceedings of Machine Learning Research},
  volume    = {70},
  pages     = {22--31},
  year      = {2017}
}

@inproceedings{Chow2018NeurIPS,
  author    = {Yinlam Chow and Ofir Nachum and Edgar Duenez{-}Guzman and Mohammad Ghavamzadeh},
  title     = {A Lyapunov-Based Approach to Safe Reinforcement Learning},
  booktitle = {Advances in Neural Information Processing Systems 31 (NeurIPS)},
  year      = {2018}
}

@inproceedings{Bouabdallah2004ICRA,
  author    = {Samir Bouabdallah and Pierpaolo Murrieri and Roland Siegwart},
  title     = {Design and Control of an Indoor Micro Quadrotor},
  booktitle = {Proceedings of the 2004 IEEE International Conference on Robotics and Automation (ICRA)},
  year      = {2004},
  volume    = {5},
  pages     = {4393--4398},
  doi       = {10.1109/ROBOT.2004.1302409}
}

@inproceedings{Bouabdallah2004IROS,
  author    = {Samir Bouabdallah and Andr{\'e} Noth and Roland Siegwart},
  title     = {{PID} vs {LQ} Control Techniques Applied to an Indoor Micro Quadrotor},
  booktitle = {2004 IEEE/RSJ International Conference on Intelligent Robots and Systems (IROS)},
  year      = {2004},
  pages     = {2451--2456},
  doi       = {10.1109/IROS.2004.1389776}
}

@article{Castillo2004TCST,
  author    = {Pedro Castillo and Alejandro Dzul and Rogelio Lozano},
  title     = {Real-Time Stabilization and Tracking of a Four-Rotor Mini Rotorcraft},
  journal   = {IEEE Transactions on Control Systems Technology},
  volume    = {12},
  number    = {4},
  pages     = {510--516},
  year      = {2004},
  doi       = {10.1109/TCST.2004.825052}
}

@inproceedings{Madani2006IROS,
  author    = {Tarek Madani and Abdelaziz Benallegue},
  title     = {Backstepping Control for a Quadrotor Helicopter},
  booktitle = {2006 IEEE/RSJ International Conference on Intelligent Robots and Systems (IROS)},
  year      = {2006},
  pages     = {3255--3260},
  doi       = {10.1109/IROS.2006.282433}
}

@inproceedings{Hoffmann2007AIAA,
  author    = {Gabriel M. Hoffmann and Haomiao Huang and Steven L. Waslander and Claire J. Tomlin},
  title     = {Quadrotor Helicopter Flight Dynamics and Control: Theory and Experiment},
  booktitle = {AIAA Guidance, Navigation and Control Conference and Exhibit},
  year      = {2007},
  note      = {AIAA Paper 2007-6461},
  doi       = {10.2514/6.2007-6461}
}

@inproceedings{Bouabdallah2007IROS,
  author    = {Samir Bouabdallah and Roland Siegwart},
  title     = {Full Control of a Quadrotor},
  booktitle = {2007 IEEE/RSJ International Conference on Intelligent Robots and Systems (IROS)},
  year      = {2007},
  pages     = {153--158},
  doi       = {10.1109/IROS.2007.4399042}
}

@article{Nagaty2013JIRS,
  author    = {Amr Nagaty and Sajad Saeedi and Carl Thibault and Mae L. Seto and Howard Li},
  title     = {Control and Navigation Framework for Quadrotor Helicopters},
  journal   = {Journal of Intelligent \& Robotic Systems},
  volume    = {70},
  number    = {1--4},
  pages     = {1--12},
  year      = {2013},
  doi       = {10.1007/s10846-012-9789-z}
}

@article{Faessler2018RAL,
  author    = {Matthias Faessler and Antonio Franchi and Davide Scaramuzza},
  title     = {Differential Flatness of Quadrotor Dynamics Subject to Rotor Drag for Accurate Tracking of High-Speed Trajectories},
  journal   = {IEEE Robotics and Automation Letters},
  volume    = {3},
  number    = {2},
  pages     = {620--626},
  year      = {2018},
  doi       = {10.1109/LRA.2017.2776353}
}

@article{Amoozgar2012IFAC,
  author    = {Mohammad Hadi Amoozgar and Abbas Chamseddine and Youmin Zhang},
  title     = {Fault-Tolerant Fuzzy Gain-Scheduled {PID} for a Quadrotor Helicopter Testbed in the Presence of Actuator Faults},
  journal   = {IFAC Proceedings Volumes},
  volume    = {45},
  number    = {3},
  pages     = {282--287},
  year      = {2012},
  doi       = {10.3182/20120328-3-IT-3014.00048}
}

@article{Prajna2007TAC,
  author    = {Stephen Prajna and Ali Jadbabaie and George J. Pappas},
  title     = {A Framework for Worst-Case and Stochastic Safety Verification Using Barrier Certificates},
  journal   = {IEEE Transactions on Automatic Control},
  volume    = {52},
  number    = {8},
  pages     = {1415--1428},
  year      = {2007},
  doi       = {10.1109/TAC.2007.902736}
}

@inproceedings{Berkenkamp2016CDC,
  author    = {Felix Berkenkamp and Riccardo Moriconi and Angela P. Schoellig and Andreas Krause},
  title     = {Safe Learning of Regions of Attraction for Uncertain, Nonlinear Systems with Gaussian Processes},
  booktitle = {2016 IEEE 55th Conference on Decision and Control (CDC)},
  year      = {2016},
  pages     = {4661--4666},
  doi       = {10.1109/CDC.2016.7798979}
}

@inproceedings{Berkenkamp2017NeurIPS,
  author    = {Felix Berkenkamp and Matteo Turchetta and Angela P. Schoellig and Andreas Krause},
  title     = {Safe Model-Based Reinforcement Learning with Stability Guarantees},
  booktitle = {Advances in Neural Information Processing Systems 30 (NeurIPS)},
  year      = {2017},
  pages     = {909--919}
}

@inproceedings{Cheng2019AAAI,
  author    = {Richard Cheng and G{\'a}bor Orosz and Richard M. Murray and Joel W. Burdick},
  title     = {End-to-End Safe Reinforcement Learning through Barrier Functions for Safety-Critical Continuous Control Tasks},
  booktitle = {Proceedings of the AAAI Conference on Artificial Intelligence},
  volume    = {33},
  number    = {1},
  pages     = {3387--3395},
  year      = {2019},
  doi       = {10.1609/aaai.v33i01.33013387}
}

@book{Bellman1957,
  author    = {Richard E. Bellman},
  title     = {Dynamic Programming},
  publisher = {Princeton University Press},
  address   = {Princeton, NJ},
  year      = {1957}
}

@book{Puterman1994,
  author    = {Martin L. Puterman},
  title     = {Markov Decision Processes: Discrete Stochastic Dynamic Programming},
  publisher = {John Wiley \& Sons},
  address   = {New York, NY},
  year      = {1994}
}

@book{SuttonBarto2018,
  author    = {Richard S. Sutton and Andrew G. Barto},
  title     = {Reinforcement Learning: An Introduction},
  edition   = {2},
  publisher = {MIT Press},
  address   = {Cambridge, MA},
  year      = {2018}
}

\end{document}